\documentclass[12pt]{amsart}
\usepackage{latexsym,amsmath,amsfonts,amscd,amssymb}
\usepackage{graphicx}
\usepackage{enumerate}
\usepackage[symbol]{footmisc}
\usepackage{xcolor}

\newtheorem{theorem}{Theorem}[section]
\newtheorem{theorem*}{Theorem}

\newtheorem{proposition}[theorem]{Proposition}
\newtheorem{proposition*}[theorem*]{Proposition}
\newtheorem{corollary}[theorem]{Corollary}
\newtheorem{corollary*}[theorem*]{Corollary}

\newtheorem{definition}[theorem]{Definition}

\newtheorem{example}[theorem]{Example}
\newtheorem{note}[theorem]{Note}
\newtheorem{note*}[theorem*]{Note}

\theoremstyle{remark}

\newtheorem{remark*}[theorem*]{Remark}

\newcommand{\EE}{{\mathbb E}}

\newcommand{\NN}{{\mathbb N}}

\newcommand{\RR}{{\mathbb R}}

\title{Profit lag and alternate network mining}


\author[C. Grunspan]{Cyril Grunspan}
\address{Cyril Grunspan\newline{}\indent L\'eonard de Vinci P\^ole Univ, Finance Lab\newline{}\indent Paris, France, }
\email{cyril.grunspan@devinci.fr}
\author[R. P\'{e}rez-Marco]{Ricardo P\'{e}rez-Marco}
\address{Ricardo P\'{e}rez-Marco\newline{}\indent CNRS, IMJ-PRG, Univ. Paris-Diderot \newline{}\indent Paris, France}
\email{ricardo.perez.marco@gmail.com}

\medskip
\address{\footnotesize Author's Bitcoin Beer Address (ABBA)\footnote{\tiny Send some anonymous and moderate satoshis to support our research at the pub.}:
1KrqVxqQFyUY9WuWcR5EHGVvhCS841LPLn} 
\address{\includegraphics[scale=0.5]{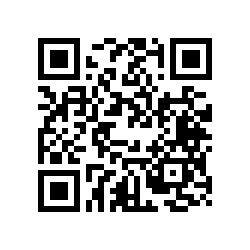}}

\begin{document}

\begin{abstract}
For a mining strategy we define the notion of ``profit lag'' as the minimum time it takes to be profitable after that moment.  
We compute closed forms for the profit lag and the revenue ratio for the strategies ``selfish mining'' and ``intermittent selfish mining''. This confirms some earlier numerical simulations and 
clarifies misunderstandings on profitability in the literature.
We also study mining pairs of PoW cryptocurrencies, often coming from a fork, with the same mining algorithm.
This represents a vector of attack that can be exploited using the ``alternate network mining'' strategy that we define. 
We compute closed forms for the profit lag and the revenue ratiofor this strategy that is
more profitable than selfish mining and intermittent selfish mining. It is also harder to counter since it 
does not rely on a flaw in the difficulty adjustment formula that is the reason for profitability of the other strategies.
\end{abstract}

\date{September 30, 2020}

\maketitle

 \begin{quotation}
  \it{Savez-vous comment on fait son chemin ici ?\footnote[1]{ Balzac, \textit{Le P\`ere Goriot.}}}
  
 \end{quotation}

\section{Introduction}
\subsection{Nakamoto consensus}
The founding paper ''Bitcoin, a peer-to-peer electronic cash system'' was announced at the end of October 2008 
on a cryptography mailing list \cite{N08-metz, N08}. A first version of the software implementing the new protocol 
developed in the article was then released on January 2009 \cite{N09}.

Bitcoin gradually enjoyed worldwide success and is today the cornerstone of the new crypto-economy 
valuated at several hundreds of billions of dollars.

Two main reasons have contributed to its success. The decentralization of the network allows to transfer value on the internet 
from user to user without the assistance of a third party. Therefore, no central server, nor bank, nor  jurisdiction can block payments \cite{N08}.
A second reason is the invention of "smart-contracts" which made bitcoin the first programmable currency \cite{A14, S19}. This allows
to construct the ``Lightning Network", on top of the Bitcoin network that offers instant and secure money transfer \cite{DP16}.

Bitcoin consensus, sometimes referred today as Nakamoto Consensus, is based on the use of proof of work (PoW). PoW was originally invented 
to fight e-mail spam \cite{DN92}. It was then used by H. Finney in the design of RPOW, the first cryptocurrency based on PoW, but was not decentralized. 
PoW plays a crucial role to secure Bitcoin protocol and in the minting process of new bitcoins. Academically speaking it combines computer security 
and classical probability theory. The main problem solved by Bitcoin and its PoW design is the prevention of double spend attacks without relying on a central server. 
Instead of searching for a flawless deterministic distributed consensus, Nakamoto's point is that, under reasonable conditions, the 
probability of success of such attacks is negligible \cite{N08,R14,GPM17,GZ19} and economically non-profitable \cite{GPM19}.

\subsection{Mining Process}
 
At any time, miners are working to build a new block from new transactions. This is
achieved solving a laborious cryptographic puzzle which involves heavy computation and use of energy. Miners iterate the calculation of 
images of a cryptographic hash function of data from a new block by varying a "nonce" (and an "extra-nonce") until 
finding a result fulfilling certain criterion.

If successful, the miner broadcasts his discovery to the rest of the network
which checks that the solution is legit. Then the new block is added to the previous known chain of blocks. The miner is then 
rewarded in newly minted bitcoins determined by the protocol and by the transaction fees of transactions included in the validated block. 
This sequence of blocks forms the  ``Bitcoin blockchain''. 
This is a secured distributed ledger recording all
past validated transactions \cite{W19}.
   
\subsection{Selfish Mining}
On the early days, the community of bitcoiners met in different internet forums, among them bitcointalk.org, a forum originally created 
by Nakamoto in November 2009. In particular, they tried to understand in depth how Bitcoin works. Some people 
had doubts of the 
specific point in the protocol that requires miners to broadcast their blocks as soon as they are validated through the  
proof-of-work \cite{RH10}. This fact allows the blockchain to record several thousand transactions every ten minutes in average. If a miner 
witholds secretly a block, he risks to loose the reward awarded to this block (coinbase and transaction fees) to a faster miner.
  
It is implicit in the founding paper that Nakamoto believed this to be the optimal strategy for miners \cite{N08}. 
Moreover, it is required in a decentralized protocol that the private economic interests of participants are in line with the protocol rules. 

However, several deviant strategies were proposed in the bitcointalk.org forum and a seminal paper by Rosenfeld examined the problem 
in 2011 \cite{R11}. Then by the end of 2013, two articles showed that other alternative strategies than the honest one can be more profitable under
suitable conditions \cite{ES14,B13}.

By modeling the progression of the blockchain with a  
Markov chain, the authors showed that a certain deviant strategy, the so-called `` selfish mining '' is more profitable in the long run than honest mining \cite{ES14}. 
This is in particular the case when a mining pool detains more than 33 \% of the whole network hashpower.

The assumptions underlying the result, such as the cohesion of mining pools, were challenged. In particular, the fact that the model considers only two groups 
of miners and only one selfish mining pool in the network \cite{F13,BC14}.

These objections are founded. For example, one could imagine that some miners participating in a mining pool decide to selfish mine not only against 
honest miners but also against their own mining pool.

Using new martingale techniques and a rigorous profitability model, it was proved in  \cite{GPM18} that the true reason for the profitability of selfish mining 
resides in a flaw in the difficulty adjustment formula. An important result obtained by these new techniques is that without difficulty adjustment, the honest
mining strategy is the most profitable one, hence vindicating Nakamoto's original belief. The flaw in the difficulty adjustment formula can be easily corrected.
%
In the presence of a block-withholding attacker, the network, that does not track the production of orphan blocks, underestimates the real hash power deployed. 
Ultimately this makes the production of blocks easier than normal, which boosts miners income. 

In particular, the selfish mining strategy only becomes profitable after a first difficulty adjustment, and after recouping all the cost employed 
reducing the difficulty (at the expense of regular 
honest mining income). At first, blocks are generated at a slower pace and all miners have their profitability reduced. 
Then after a first difficulty adjustment, blocks arrive faster and the selfish miner profits.

This theoretically explains why the strategy has never been implemented. No miner risks to mine at a loss for several weeks while assuming that mining 
conditions remain the same i.e., one selfish miner and same overall hash power over time. Moreover, if a miner has substantial hashpower, then probably 
double spend attacks based on 51\% hashrate dominance  are easier to achieve by increasing the hashpower.

The original selfish mining strategy assumes also the non-realistic hypothesis that there is no arrival of new miners to the network attracted by the 
lower mining difficulty. 



\subsection{Smart mining}
As explained, the advantage of selfish mining is based on reducing the difficulty  
but continue to profit from block rewards, even before the first difficulty adjustment. 
Another idea presented in \cite{GS19} consists in withdrawing temporarily 
from the network to lower the difficulty. Then come back after a difficulty adjustment 
to take advantage of the lower hash rate mining.
The authors call such a strategy ``smart mining" (sic).
By considering the cost of mining (fixed cost and variable cost per time unit), it has been shown that this strategy can be more profitable than honest mining, 
even for low hashrates. In some sense, this attack is more serious than the selfish mining 
strategy because it does not exploit a particular flaw of the difficulty adjustment formula. 

\subsection{Intermittent selfish mining}
Another possible strategy is the strategy of ''Intermittent selfish mining`` (ISM) that 
consists in alternating selfish and honest mining during consecutive difficulty periods.
This was early discussed in social media and then numerically studied in \cite{ NRS20}.

The idea is to fully profit from the decrease of the difficulty but the downside is that 
the difficulty does not stabilize and recovers
after each phase of honest mining.

\subsection{Alternate network mining}

In this article we introduce and study a strategy similar to "smart mining" that we name ''alternate network mining`` (ANM).  
The difference is that when the miner withdraws from the network he goes on to mine in another 
network with similar mining algorithm. Clearly, if honest mining is profitable, ''alternate network mining`` is more profitable than 
''smart mining``, and we prove that it is also more profitable than ''selfish mining`` and ''intermittent selfish mining``. 
It is also effective for low hashrates. Therefore, the existence of cryptocurrencies operating with identical 
PoW algorithms represent a vector of attack for both networks (this was already observed in \cite{GPM18}).
Our study focus mainly on Bitcoin and Bitcoin Cash, but can be adapted to other pairs such as BCH / BSV or ETH / ETC.

\subsection{Organization of this article}

We start briefly recalling the mathematics of Bitcoin  mining
and the profitability model for comparing mining strategies \cite{GPM18,GPM20-a}. Then, we review the selfish mining 
strategy and show the equivalence between the Markov chain approach and the martingale
approach for the computation of profitabilities \cite{ES14, GPM18}. Also we review how to compute time before profit with both approaches.


Then, we turn to ISM strategy and compute in closed form the apparent hashrate 
of the strategy as well as the time before profit. These results are new since before only numerical simulations were available. 
In passing we correct false claims and misunderstanding of the authors of \cite{NRS20}.
Then, we compute the profitability of the ANM strategy. The present article is self contained.

\section{Modelization}

\subsection{Mining and difficulty adjustment formula}
From new transactions collected in the local database (the ``mempool''), a miner builds a block $B$ containing
a trace of a previous block and a variable parameter (called a ``nonce") until he finds a solution to the inequality 
$h(B) < \frac{1}{\delta}$ where $\delta$ is the difficulty of the hashing problem and $h$ is a cryptographic 
hash function ($h = SHA_{256}\circ SHA_{256}$ for Bitcoin). This solution is the ''proof-of-work`` 
because it proves, under the usual assumptions on the pseudorandom properties of the hash function, 
that the miner has spend energy.   
For Bitcoin, the difficulty is adjusted every 2016 blocks, which
corresponds on average to 2 weeks. For Bitcoin Cash and for for Bitcoin SV, the difficulty is adjusted at each block using timestamps of 
the previous $144$ blocks (about one day). For Ethereum, the difficulty is adjusted at each block using the lapse for the discovery of the 
previous block. Ethereum difficulty adjustment algorithm was modified in June 2017
\cite{B17}. Since then, the adjustment takes into account the production of special orphan blocks
also called "uncles". These are orphan blocks that are not too
far from the "parent block". Ethereum Classic did not implement this new difficulty adjustment formula, 
and is more vulnerable to selfish mining \cite{FN19, GPM20-b}

\subsection{Notations}\label{notations}
In the article, we use Satoshi's notations from his founding paper. Thus, we denote by $p$ (resp. $q=1-p$) the 
relative hash power of the honest miners (resp. the attacker). So, at any time, $q$ 
is the probability that the attacker discovers a new block before the honest miners. We also denote by $N(t)$ 
(resp. $N'(t)$) the counting process of blocks validated by the honest miners (resp. the attacker) 
during a period of $t$ seconds from an origin date $0$.
The hash function $SHA_{256}$ is assumed to be perfect, the time it takes to find a block follows an exponential 
distribution and hence $N(t)$ (resp. $N'(t)$) is a Poisson process. Initially (before a difficulty adjustment) 
the parameter of the Poisson process $N(t)$ (resp. $N'(t)$) is $\frac{p}{\tau_0}$ (resp. $\frac{q}{\tau_0}$)
with $\tau_0 = 600$ seconds. We also denote by $S_n$ (resp. $S'_n$) the time it takes for the honest miners (resp. the attacker) 
to mine $n$ blocks, and $\tilde{S}_n$ the time it takes for the whole network to add $n$ blocks to the official 
blockchain. The random variables $S_n$ and $S'_n$ follow Gamma distributions. 
The same occurs for $\tilde{S}_n$ when the attaker mines honestly (see \cite{GPM17}). The multiplicative difficulty adjustment parameter is 
denoted by $\delta$. It is updated every $n_0=2016$ blocks. Thus, $\tilde{S}_{n_0}$ represents a complete mining 
period of $n_0 $ official blocks and at this date, the protocol proceeds with a new difficulty.

We use the notation $\gamma$ 
for the connectivity of the attacker. 
In case of a competition between two  blocks, one of which has been mined and published by the attacker, 
$\gamma$ represents the fraction of honest miners mining on top of the attacker's block.

\subsection{Profitability of a mining strategy}
We consider a miner which is active over a long period compared to the average mining time of a block. What counts 
for his business is its Profit and Loss (PnL) per unit of time. We denote by $R(t)$ the total income of the miner 
between time $0$ and time $t>0$. Similarly, we denote by $C(t)$ the total cost he incurs during this period. 
Note that $C(t)$ is not restricted to direct mining costs but also includes all expenditures. 
So, his PnL is $PnL(t)=R(t)-C(t)$ and seeks to maximize $\frac{PnL(t)}{t}=\frac{R(t)}{t}-\frac{C(t)}{t}$ for $t\to \infty$. 
We assume that the mining cost is independent of the mining strategy. Indeed,
whether the miner broadcasts his blocks or keeps them secret, it has no impact 
on its mining costs. Also, whether he mines a certain
cryptocurrency or another with the same hash algorithm, 
this does not change the mining cost per unit of time. This last quantity
essentially depends on electricity costs, price of his machines, 
salaries of employees, etc. In particular, when the miner mines at full regime, 
then the cost of mining per unit of time does not depend on the strategy.
In this situation, the relevant quantity is the revenue ratio $\Gamma= \lim_{t\to +\infty}\frac{R(t)}{t}$.
A strategy $\xi$ is more profitable than a strategy $\xi'$ if and only if its revenue ratio is greater: 
$\Gamma\bigl({\xi}\bigr)\geq \Gamma\bigl({{\xi}'}\bigr)$
(see Corollary \ref{thecoro} below).

\begin{note}
The literature is often obscure about profitability model and sometimes disregards mining costs without proper justification. 
For example, in \cite{NRS20} one can read: ``We omit transaction fees and mining costs from our analysis''. 
But, as proved in \cite{GPM17}, the only necessary assumption in order to compare strategies is to have the 
mining cost per unit of time to be the same for both strategies.
Otherwise, the profitability analysis without considering costs does not make sense. Thus, in the smart mining strategy 
considered in \cite{GS19}, the authors are naturally led to consider mining costs which depend on the miner being active or not. 
In the alternate network mining strategy considered in Section \ref{ANM}, the miner never remains inactive assuming that the mining 
costs are equal in both networks.
\end{note}

\subsection{Attack cycles}
A strategy consists in  attack cycles. The end of a cycle is determined by a 
{\it stopping time}. At the start of a cycle, the attacker and the honest miners have the same view 
of the official blockchain and mine on top of the same block. In general, thought this is not mandatory, during the attack cycle, 
the attacker mines on a fork that he keeps secret.

\begin{example}
The sequence SSSHSHH corresponds to a particular attack cycle for the selfish mining strategy: the attacker first mines three blocks
in a row that are kept secret (blocks "S"); then the honest miners mine one (block "H"); then the attacker mines another one 
(still secret); then the honest miners mine two blocks in a row and so the attacker decides to publish his entire fork 
because he only has a lead of 1 on the official blockchain. The attack cycle then ends. In this case, the attacker is victorious: 
all the blocks he has mined end up in the official blockchain.
\end{example}

The attacker iterates attack cycles. It follows that by noting $R_i$ the miner's income after a $i$-th cycle and by $T_i$ the duration time of this cycle,
the revenue ratio of the strategy is equal to 
$\frac{\sum_{i=1}^{n} R_i}{\sum_{i=1}^{n} T_i}$ for $n\to +\infty$. This quantity converges to $\frac{\EE[R]}{\EE[T]}$
by the strong law of large numbers, provided that the duration time of a cycle
attack $T$ is integrable and non-zero. Likewise, assuming that costs are integrable, the cost of mining by time unit converges to $\frac{\EE[C]}{\EE[T]}$.
Thus we can state the following corollary.

\begin{proposition}
Let $\xi$ and $\xi'$ be two mining strategies. Let $R$ and $R'$ (resp. $C$ and $C'$) be the revenue (resp. cost) of the miner by attack cycle. We denote also by $T$ and $T'$ the duration times of attack cycles for $\xi$ and $\xi'$. Then, $\xi$ is more profitable than $\xi'$ if and only if 
$\frac{\EE[R_{\xi}]}{\EE[T]} - \frac{\EE[C_{\xi}]}{\EE[T]}
> \frac{\EE[R_{{\xi}'}]}{\EE[T]} - \frac{\EE[C_{{\xi}'}]}{\EE[T]}$
\end{proposition}

\begin{corollary}\label{thecoro}
If moreover we assume that the cost of mining per unit of time of $\xi$ and $\xi'$ are equally distributed, then $\xi$ is more profitable than $\xi'$ if and only if 
$\frac{\EE[R_{\xi}]}{\EE[T]} > \frac{\EE[R_{{\xi}'}]}{\EE[T]}$.
\end{corollary}

\begin{example}
The attack cycles for the honest strategy are simply the time taken by the whole network to discover a block. Consequently, when the miner mines honestly, his revenue ratio is $\frac{q b}{\tau_0}$
where $b$ is the mean value of a block (coinbase and transaction fees).
\end{example}

\subsection{Performant strategy and profit lag}\label{claritbf}
\begin{definition}
A mining strategy is a performant strategy if its revenue ratio is greater than the revenue ratio of the honest strategy, i.e. 
$\frac{\EE[R(T)]}{\EE[T]} > \frac{q b}{\tau_0}$, where $R(T)$ is the revenue per attack cycle and $T$ is the duration time of an attack cycle.
\end{definition}

For a performant strategy $\xi$, we have 
$\EE[R(T)]\geq q b \frac{\EE[T]}{\tau_0}$.  So, at $T$-time, the miner is better off, on average, following strategy $\xi$ 
that mining honestly. However, by continuing the strategy $\xi$, 
it is possible that the miner will incur in looses for a longer period of time. Nothing prevents having
$\EE[R(\tau)]< q b \frac{\EE[\tau]}{\tau_0}$ for a certain stopping time $\tau>T$. This, indeed, happens  
for the ISM strategy. Below we define the notion of profit lag of a mining strategy.

\begin{definition}
Let $\xi$ be a mining strategy and $\tau$ a stopping time. 
\begin{itemize}
\item $\xi$ is profitable at date $\tau$ if 
$\EE[R(\tau)]\geq q b \frac{\EE[\tau]}{\tau_0}$. 

\item $\xi$ is definitely a performant strategy at date $\tau$ if
$
\EE[R(\tau')]\geq q b \frac{\EE[\tau']}{\tau_0}
$
for any stopping time $\tau'>\tau$ a.s. 

\item The profit lag is the smallest stopping time $\tau$ with this property.
\end{itemize}
\end{definition}

This definition is sound since the infimum of stopping times is a stopping time. 
Note that if $\xi$ is a performant strategy with duration time $T$ for an attack cycle, 
then $\xi$ is profitable at a date  $T$.  
But in general, $\xi$ is \textit{not} 
definitely profitable at date $T$: the profit lag is in general longer than $T$. 
However we can prove that if $\xi$ is a performant strategy then the profit lag is finite.

\begin{proposition}
Let $\xi$ be a performant mining strategy. Then, there exists $\tau$
a stopping time such that $\xi$ is definitely a performant strategy at date $\tau$.
\end{proposition}

\begin{proof}
We keep the same notations as above. 
Let $X(t)=\EE[R(t)-\frac{q b t}{\tau_0}]$ for $t\in\RR_+$.
Since $\xi$ is a performant repetitive strategy, we have $X(T)>0$ and 
$$\text{Inf}\, 
\lbrace X(\tau)\, ;\, \tau > 0 \text{ stopping time} \rbrace
=\text{Inf}\, 
\lbrace X(\tau)\, ;\, \tau \text{ stopping time}\in[0,T]\rbrace
>-\infty
$$
Let us denote by $m\in\RR_{-}$ this quantity and let $n$ be an integer with $n>\frac{|m|}{X(T)}$. Then, if $\tau > nT$ is a stopping time, 
we have:
$$
X(\tau)=X(n T) + X(\tau - n T) \geq n X(T) + m > 0
$$
Hence, we get the result.
\end{proof}

\section{Selfish mining revisited}
Selfish mining strategy can be described by the stopping time which defines the end of an attack cycle.

\begin{definition}
The end of an attack cycle for the selfish mining strategy is given by the stopping time $T={\text{\rm Inf }}\bigl\lbrace t>S_1\, /\, 
N(t) = N'(t) -1 +2\cdot 1_{S_1<S'_1} +2\cdot 1_{S'_1<S_1<S_2<S'_2}\bigr\rbrace$.
\end{definition}

\begin{example}
In the case considered above, $S'_1< S'_2 <S'_3<S_1 <S'_4 < S_2<S_3$, we have
$1_{S_1<S'_1} = 1_{S'_1<S_1<S_2<S'_2} = 0$ and 
$T={\text{\rm Inf}}\,\bigl\lbrace t>S_1\, /\, N(t) = N'(t) -1\bigr\rbrace = S_3$.
\end{example}

We note also by $L$ the number of blocks added to the official blockchain after an attack cycle. It is clear that $|N(T)-N'(T)|=1$ so $L=\frac{N(T)+N'(T)+1}{2}$. In the previous example, we have $N(T)=3$, $N'(T)=4$ and $L=4$.

\begin{proposition}\label{RTL}
We have:
\begin{align*}
\EE [R] &=  \frac{(1+pq)(p - q) + pq}{p - q}\, qb -(1-\gamma) p^2q  \, b \ ,\\
\EE [T]  &= \frac{(1+pq) (p - q) + pq}{p - q} \, \tau_0  \, \\
\EE[L] &= \frac{1+p^2 q+p-q}{p-q} b \\
\end{align*}
\end{proposition}

\begin{proof}
The counting process $N(t)$ (resp. $N'(t)$) is a Poisson process. So, $T$ (resp. $T'$)
is an integrable stopping time. Then the result follows from Doob stopping time theorem applied to the compensated martingales $N(t)-\alpha t$ 
(resp. $N'(t)-\alpha' t$) with $\alpha = \frac{p}{\tau_0}$ 
(resp. $\alpha' = \frac{q}{\tau_0}$) and to the finite stopping time $T\wedge t$
for a fixed $t>0$. Then, we take the limit when $t\rightarrow\infty$
using the monotone convergence theorem. We get as limits 
$\EE[T]$, $\EE[N(T)]$, $\EE[N'(T)]$ and $\EE[L]$. 
Finally, we also get $\EE[R]$ by observing that 
$\EE[R]=\EE[N'(T)] b-(1-\gamma) p^2 q b$ since the only attack cycle when 
$R\not= N'(T)$ happens when the attacker mines first a block, then the 
honest miners mine one block and the honest miners mine another block on top of the honest block.
\end{proof}

We set the date $t=0$ when the selfish miner starts his attack, at the beginning of a new period of $2016$ official blocks, just after a difficulty adjustment. 

\begin{corollary}\label{Sn0}
Let $\delta$ be the first difficulty parameter.
We have 
$\EE[\delta]=\frac{p-q+p q(p-q) + p q}{p^2 q + p-q}$
and $\EE[\tilde{S}_{n_0}]=\EE[\delta]\cdot n_0\cdot\tau_0$
\end{corollary}
\begin{proof}
Let $\nu$ be the number of cycles before we get a difficulty adjustment. By Wald theorem, we get
$\EE[L]\cdot \EE[\nu]= n_0$ and
$\EE[\tilde{S}_{n_0}]=\EE[\nu]\cdot \EE[T]$. So,
$\EE[\delta]=\frac{\EE[\tilde{S}_{n_0}]}{n_0 \tau_0}=\frac{n_0 \EE[T]}{\EE[L] \tau_0}$. Hence, we get the result using Proposition \ref{RTL}
\end{proof}

We can now calculate the apparent hashrate of the strategy which is defined as the fraction of blocks of the attacker in the official blockchain \textit{on the long term}.

\begin{corollary}\label{qprime}
The apparent hashrate of the selfish mining strategy is 
$$
q'=\frac{((1 + pq) (p - q) + pq) q - (1 - \gamma) p^2 q (p - q)}{p^2 q + p - q}
$$
\end{corollary}

\begin{proof}
After a difficulty adjustment, the mining difficulty is divided by $\delta$. 
So the parameters of the Poisson processes $N(t)$ and $N'(t)$ are each multiplied by $\delta$ and also the duration times of attack cycles are divided by $\delta$. On the other hand, the revenue of the selfish miner remains constant after each attack cycle since the probability to be first to discover a new block remains the same. 
Therefore, his new revenue ratio after difficulty adjustment is 
$q'\frac{b}{\tau_0}$ with
$q'=\left(\frac{\EE[R]}{\EE[T]}\cdot \delta\right)\cdot \frac{\tau_0}{b}$. 
Hence, we get the result using Proposition \ref{RTL}.
\end{proof}
By rearranging terms, we observe that $q'$ is nothing but the quantity $R_{pool}$ from \cite{ES14}. But, more importantly, we have proven the following proposition.

\begin{proposition}\label{arrivalsSM}
Before a difficulty adjustment, official blocks are discovered on average every period of 
$\delta\cdot\tau_0$. The revenue ratio of the selfish miner is 
$\frac{q'}{\delta}\frac{b}{\tau_0}$ (lower than $q\frac{b}{\tau_0}$). 
After a difficulty adjustment, officials blocks are discovered on average every period of $\tau_0$.
The revenue ratio of the selfish miner equal $q'\frac{b}{\tau_0}$ (which is now greater than $q\frac{b}{\tau_0}$).
\end{proposition}

\subsection{Previous state-machine approach revisited}
We can model the different states taken by the network with the help of a Markov chain $(X_n)$ as was done in \cite{ES14}. The set of states is 
$\NN\cup\{0'\}$. The state 
$\{0'\}$ corresponds to the state when there is a competition between two blocks in the network. The state 
$\{n\}$ corresponds to the state when the attacker has a lead of $n$ blocks over the honest miners.
In this framework, the duration time of an attack cycle is proportional to the first return time 
$\nu={\text{Inf}}\,\left\lbrace n\in \NN^{*}\, / \, X_n = 0\right\rbrace$.
It turns out that the Markov chain is irreducible and all states are positive recurrent. 
It has naturally a stationary probability $(\pi_n)$ and the computation shows that  $\pi_0 = \frac{1-2q}{1-4q^2+2q^3}$ \cite{ES14}.
Therefore, a classical probability result on Markov chains shows that $\EE[\nu] = \frac{1}{\pi_0}$ \cite{DMPS20}. Then
it is easy to check that 
$$\EE[\nu]= \frac{(1+pq) (p - q) + pq}{p - q}$$ 
We recover the formula obtained for $\EE[T]$ in Proposition \ref{RTL}. Also, for $r_{pool}$ and $r_{others}$ 
defined in \cite{ES14} we have
$$
\frac{1}{r_{pool}+r_{others}}=\frac{1-4 q^2+2 q^3}{1-q-2 q^2+q^3}=\EE[\delta]
$$
So we get the difficulty adjustment parameter given in Proposition
\ref{Sn0}. Finally, the revenue ratio before difficulty adjustment is
$\frac{\EE[R]}{\EE[T]}=r_{pool}\cdot\frac{b}{\tau_0}$.
The formula given in \cite{ES14} for the apparent hashrate 
$q'=\frac{r_{pool}}{r_{pool}+r_{others}}$ is the same that we find:
$r_{pool}$ is the revenue ratio before difficulty adjustment and 
$\frac{1}{r_{pool}+r_{others}}>1$ is the (mean) difficulty adjustment parameter $\delta$.
For the purposes of computation of the profitability there are two equivalent models: 
one with Poisson processes
(as considered first by Satoshi Nakamoto in his founding paper and continued in \cite{GPM18}) 
and another with Markov chains. They are not equivalent for other purposes, as for example 
to show the optimality of honest mining when there is no difficulty adjustment.
Note also that there exists a third very effective combinatorical model
that captures directly the long term behavior of 
a deviant strategy. For the selfish mining strategy, stubborn strategy and trailing strategy, it uses the combinatorics of Dyck words and Catalan distribution 
(see \cite{KMNS16, GPM18b, GPM18e}).

\subsection{Profit lag.}
We recall that according to notations in Section \ref{notations},
$\tilde{S}_{n_0}$ denotes the time used by the network to reach the first difficulty adjustment. 
\begin{proposition}\label{tbfsm}
Set $t_0=\frac{q\delta - q'}{q'-q}n_0\tau_0$.
On average, the selfish mining strategy is not profitable before $\tilde{S}_{n_0} +t_0$ and is profitable after this date.
\end{proposition}

\begin{proof}
We denote by $\Delta(t)$ the difference between the revenue of the selfish miner at date $t$ and 
the revenue of honest mining. 
As we have seen, the revenue ratio of the selfish miner before the first 
difficulty adjustment is lower than the revenue ratio of the honest strategy.
So, the selfish mining strategy cannot be profitable before this date.
Moreover, using Proposition \ref{arrivalsSM}, we get: 
$\EE[\Delta(\tilde{S}_{n_0})]=(\frac{q'}{\delta}-q) n_0\delta b
= -(q \delta - q') n_0 b <0$ and for $t>0$,
we have
\begin{equation*}
\EE [\Delta(\tilde{S}_{n_0}+t)] =  (q'-q)\frac{t}{\tau_0}b +\EE [\Delta(\tilde{S}_{n_0})]= (q'-q)\frac{t}{\tau_0}b -(q \delta - q') n_0 b = (q'-q) (t-t_0) \frac{b}{\tau_0}
\end{equation*}
\end{proof}


\begin{figure}[!ht]\label{FigSM}
   \includegraphics[height=5cm, width=10cm]{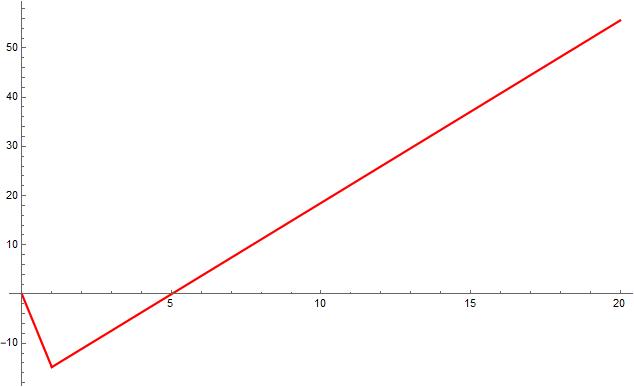}
   \caption{
   \footnotesize{Difference of average revenue between the selfish miner and honest 
   miner for $q=0.1$ and $\gamma=0.9$. $X$-axis: progression of the official blockchain, 
   in difficulty adjustment units. 
   $Y$-axis: revenue of the miner in coinbase units.
   The strategy definitely not profitable before 5 difficulty adjustments.
   }}
\end{figure}

\begin{example}\label{0109}
When $q=0.1$ and $\gamma=0.9$, the computation shows that we have exactly 
$\frac{q\delta - q'}{q'-q}=4$. Therefore, the selfish mining strategy 
becomes definitely profitable at the date of the fifth difficulty adjustment 
which is roughly $10$ weeks of mining (it is exactly $10.1839$ weeks)
as asserted in \cite{GPM18}. Before this date, the strategy is not profitable.  
\end{example}

\begin{note}
In the introduction of \cite{NRS20} the authors claim 
(with no justification) that the previous computation from \cite{GPM18} is erroneous. 
However, in their simulations, they obtain $\EE[\tilde{S}_{n_0}]=15.29$ days which agrees, 
up to rounding, with 
the exact theoretical value from Corollary \ref{Sn0} which is $15.2872$. 
And from here, it is not hard to get Proposition \ref{tbfsm} 
as we did above. Obviously, our theoretical exact results agree  with 
selfish mining simulators that one can find online \cite{K19}. 
So, apparently, the authors of \cite{NRS20} are confused, not only about this, but also about the profitability of the Intermittent Selfish Mining (ISM)
that they analyze only numerically. In the next section we give rigorous analysis of ISM with  closed form formulas.
\end{note}


\section{Intermittent selfish mining strategy}
This strategy is composed by attack cycles with two distinct phases. 
Each phase corresponds to a complete period of mining of $n_0$ official blocks until 
the next difficulty adjustment. 
In the first phase, the attacker follows the selfish mining strategy. In the second phase, 
he mines honestly. 

\begin{proposition}\label{RTISM}
The mean duration time of the first (resp. second) phase is $n_0\, \tau_0\, \delta$
(resp. $\frac{n_0\, \tau_0}{\delta}$). During the first (resp. second) phase,
the revenue ratio of the miner is $\frac{q'}{\delta}\cdot\frac{b}{\tau_0}$ 
(resp. $q\, \delta\cdot \frac{b}{\tau_0}$) where $q'$ and $\delta$ are given in Corollary \ref{Sn0} and Coreollary \ref{qprime} (denoting by $\delta$ its expected value).
Moreover, denoting by $R$ (resp. $T$) the revenue (resp. duration time) at the end of an attack cycle, we have:
\begin{align*}
\EE [T]&= \left(\delta + \frac{1}{\delta}\right) n_0 \tau_0 \\
\EE [R]&= (q'+q) n_0 b
\end{align*}
\end{proposition}

\begin{proof}
The duration time and the revenue ratio of the attacker during the selfish mining phase has been computed in the previous section. After a first difficulty adjustment, the mining difficulty is divided by $\delta$. So, blocks arrive on an average period  
$\frac{\tau_0}{\delta}$ and the duration time of the second phase is
$\frac{n_0\, \tau_0}{\delta}$. Each time, there is a probability $q$ that the attacker is the first one to discover the next block. 
Hence, during the second phase, his revenue ratio is 
$\frac{q\, b}{\frac{\tau_0}{\delta}}=\frac{q\, \delta b}{\tau_0}$.
Finally, the revenue of the ISM miner at the end of an attack cyle is
$$
\EE[R]=\left(\frac{q'}{\delta}\frac{b}{\tau_0}\right)\,\cdot n_0\,\tau_0\,\delta+
\left(\frac{q\, \delta\, b}{\tau_0}\right)\cdot\,\frac{n_0\tau_0}{\delta}
$$
and the result follows.
\end{proof}

\begin{corollary}\label{qsecond}
The apparent hashrate of the Intermittent selfish strategy is  $q"=\frac{q+q'}{\delta+\frac{1}{\delta}}$, and
we have the closed form formula
$$
q"=\frac{q \left(1-4 q^2+2 q^3\right) 
\left(1+\gamma+(3-4\gamma)q+(5\gamma-11)q^2+(5-2\gamma) q^3\right)}
{2-2q-11q^2+10 q^3+18 q^4-20 q^5+ 5q^6}
$$
\end{corollary}

\begin{example}
This formula checks well with simulations from \cite{NRS20}.
The threshold ISM/HM is obtained when $q"=q$. When $\gamma = 0$, we get
$q"=0.365078$. The authors from \cite{NRS20} obtain $0.37$ with their simulations. 
\end{example}

\begin{corollary}
When ISM is more profitable than honest mining, 
then SM is more profitable than ISM.
\end{corollary}

\begin{proof}
By hypothesis, we have $\frac{q+q'}{\delta+\frac{1}{\delta}}>q$, which implies
$$
\frac{q'}{q}>\delta + \frac{1}{\delta}-1>1>\frac{1}{\delta + \frac{1}{\delta}-1}
$$
since $\delta+\frac{1}{\delta}>2$ for $\delta >0$. So,
$\frac{q+q'}{\delta+\frac{1}{\delta}}<q'$ and 
ISM is less profitable than SM.
\end{proof}

In Figure 2 regions in  $(q, \gamma) \in [0, 0.5] \times [0, 1]$ are colored according to which strategy is more profitable (HM is the honest mining strategy).

\begin{figure}[!ht]
   \includegraphics[height=5cm, width=10cm]{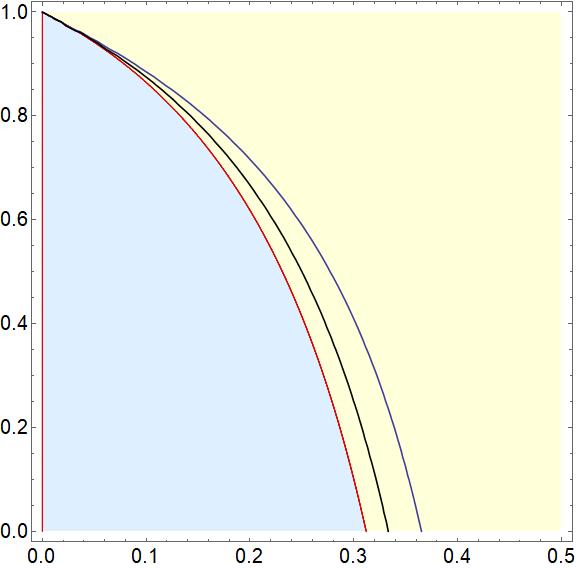}
   \caption{
   \footnotesize{Dominance regions in parameter space $(q,\gamma)$.
   The threshold SM/HM (resp. ISM/HM, resp. SM/ISM) in black (resp. blue, red).
   When ISM is more profitable than HM 
   then SM is always more profitable than ISM.
   }}
\end{figure}

\subsection{Profit lag.}
As before, we denote by $\Delta$ the difference between the average revenue of selfish and honest mining. 
According to Proposition \ref{RTISM}, after $n$ attack cycles i.e., alternatively $n$ phases of 
selfish mining and $n$ phases of honest mining, the attacker earns
$q"\cdot\left(\delta+\displaystyle\frac{1}{\delta}\right) n\cdot\,n_0\, b$ during 
$\left(\delta+\displaystyle\frac{1}{\delta}\right) n\cdot n_0 \tau_0$. So, at this date,
we have 
$$\Delta=(q"-q)\cdot\left(\delta+\frac{1}{\delta}\right)\cdot n\,n_0\, b
$$
After this date, if the attacker selfish mines again, then at the end 
of this phase of duration $\delta n_0\tau_0$, his revenue is
$q"\cdot\left(\delta+\displaystyle\frac{1}{\delta}\right) n\cdot\,n_0\, b
+\frac{q'}{\delta}\cdot n_0\,\delta\,\tau_0$. At this date, we have

\begin{align*}
\Delta &=  (q"-q)\cdot\left(\delta+\frac{1}{\delta}\right)\cdot n\,n_0\, b
+ \left(\frac{q'}{\delta}-q\right)\cdot n_0\,\delta\,b\\
 &=  (q"-q)\cdot\left(\delta+\frac{1}{\delta}\right)\cdot n\,n_0\, b
- \left(q\delta-q'\right)\cdot n_0\,b\\
\end{align*}

By the end of an attack cycle, we have on average 
$\Delta = (q"-q)\,n_0 \left(\delta+\displaystyle\frac{1}{\delta}\right) b$ which is positive when ISM is more profitable than honest mining. 
Thus when $q">q$, ISM is more profitable than HM before
a second difficulty adjustment as noticed in \cite{NRS20}. This comes as no surprise since 
the duration time of an attack cycle of 
the ISM strategy corresponds to a period of $2\times 2016$ official blocks, 
and any performant mining strategy is always profitable at the end of an attack cycle. 
See the discussion at the end of Subsection \ref{claritbf}.

The authors \cite{NRS20} fail to understand that 
after another difficulty adjustment the revenue falls because the difficulty 
increases and ISM becomes less profitable than HM. It is only after several difficulty adjustment 
periods of $2016$ blocks 
that ISM can become definitely more profitable than HM. For instance, when $q=0.1$ and 
$\gamma=0.9$ this only happens approximatively after the $13$-th difficulty adjustment. 
Therefore, on average, it takes more than 6 months for the strategy to be definitely 
profitable, and this is much longer than for classical SM. 
Obviously the authors of \cite{NRS20}
are overall confused with a naive notion of profitability and don't follow a rigorous profitability model.

It is important to understand what happens. First, the selfish miner invests in lowering the difficulty, 
and, at any moment, he can get reap immediate profits, even before the profit lag, by just mining honestly. 
But if the miner wants to repeat the attack cycle, he will need to burn again these profits for the purpose of lowering 
the difficulty. This is what happens in the first cycles of the ISM strategy. 
In Figure 3 we have the plot of the progression of $\Delta$. 

\begin{figure}[!ht]
   \includegraphics[height=5cm, width=10cm]{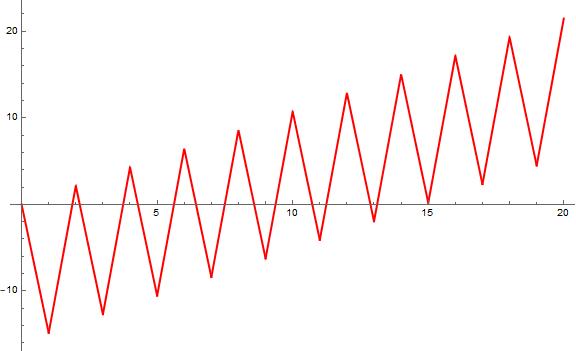}
   \caption{
   \footnotesize{
   Difference of average revenue between intermittent selfish and honest mining for $q=0.1$ and $\gamma=0.9$.
$X$-axis: progression of the official blockchain, in difficulty adjustment units. 
   $Y$-axis: revenue of the miner in coinbase units.
   The strategy is not definitely profitable before $13$ difficulty adjustments. By comparison, SM is definitely profitable after $5$ difficulty adjustments (compare with Figure 1).
   }}
\end{figure}

\section{Alternate network mining strategy}\label{ANM}
The strategy is described for alternate mining between BTC and BCH networks, but 
obviously it applies
to any pair of networks with the same mining algorithm (with frictionless switching mining operations).
This time we consider three distinct types of 
miners: the first ones mine on Bitcoin only, the second ones on Bcash (BCH), and the third 
ones alternate between Bitcoin and Bcash. The ones mining alternatively between BTC and BCH 
are labelled ``attackers'' (although his strategy is legit and respects 
both network protocols).
We assume that the revenue ratios of mining honestly on Bitcoin and Bcash are the same. 
This is approximativelly what is observed in practice since any divergence justifies a migration of hashrate from one network to the other. 
We denote by $\rho$ this common value for the attacker.
We assume also that there are no other miners and all miners mine with full power. 
In particular, the total hashrate remains constant.
The attacker starts mining on Bitcoin at the beginning
of a difficulty adjustment on BTC. An attack cycle is made of two phases. During Phase $1$, the attacker withdraws 
from BTC and mines on BCH until $n_0$ blocks have been mined on Bitcoin. During Phase 2, the attacker comes back to mine on the BTC network
until a new difficulty adjustment. We call this strategy: \textit{Alternate Mining Strategy}. 
This is a variation of smart mining. The only difference being that the miner does not remain idle during Phase 1 but goes on to mine on BCH. 
Note that this mining strategy does not provoke periodic reorganizations of the blockchain. 
The main annoyance for users of both networks is that difficulty does not stabilize and 
for Bitcoins users blocks arrive regularly at a slower pace than in the steady regime with the 
miner fully dedicated to the BTC network.

We denote by $\delta$ the difficulty adjustment parameter after Phase 1.
Since the miner comes back on BTC in the second phase, the second 
difficulty adjustment at the end of the second phase is $\frac{1}{\delta}$.

\begin{proposition}\label{rnams}
The duration phase of Phase 1, resp. Phase $2$, 
is $n_0\tau_0\delta$, resp. $n_0\tau_0\frac{1}{\delta}$.
During Phase 1, resp. Phase $2$, the revenue ratio of the attacker is
$\rho$, resp. $\rho\delta$.  
The revenue ratio of the alternate mining strategy is
$\frac{1+\delta}{\delta+\frac{1}{\delta}}\rho$.
\end{proposition}

\begin{proof}
By definition of the difficulty adjustment parameter,
the duration time of Phase 1, resp. Phase 2,  is $\delta n_0\,\tau_0$, 
resp. $\frac{1}{\delta} n_0\,\tau_0$. So the duration time of an attack cycle is
$\left(\delta+\frac{1}{\delta}\right)\,n_0\tau_0$.
During Phase 1, the attacker's revenue ratio is $\rho$ because
we assume that the attacker mines honestly during this phase with the assumption that the revenue ratio is the same for BCH and BTC. During Phase 2, the mining difficulty is divided by $\delta$. So, the revenue ratio of the attacker during this period is $\delta\cdot\rho$.
Therefore, the revenue of the attacker after an attack cycle is
$\rho\cdot n_0\,\tau_0\,\delta+(\rho\,\delta)\cdot\frac{n_0\tau_0}{\delta}$.
\end{proof}
\begin{corollary}
The alternate mining strategy is always more profitable than honest mining and selfish mining for all values of $q$.
\end{corollary}

\begin{proof}
The first statement results from $\delta>1$ that implies 
$\frac{1+\delta}{\delta+\frac{1}{\delta}}>1$.
To prove, the second statement, we remark that in Phase 1, blocks are 
only validated by honest miners. 
So this phase lasts on average $\frac{n_0 \tau_0}{p}$ and the 
difficulty parameter is updated accordingly: $\delta = \frac{1}{p}$ with
$p=1-q$. 
So, the revenue ratio of the attacker is $\frac{2-q}{2-2q+q^2}\frac{b}{\tau_0}$ as follows by replacing $\delta$ with $\frac{1}{p}$ in the formula from Proposition \ref{rnams}.
On the other hand, in the most favorable case (when $\gamma = 1$), 
the revenue ratio of the selfish miner is 
$\frac{q \left(2 q^3-4 q^2+1\right)}{q^3-2 q^2-q+1}\frac{b}{\tau_0}$.
Then, the result comes from $\frac{2-q}{2-2q+q^2} > \frac{q \left(2 q^3-4 q^2+1\right)}{q^3-2 q^2-q+1}$
for $0<q<0.5$ that we prove studying the polynomial 
$$
(2-q)\cdot (q^3-2 q^2-q+1)
-(2-2q+q^2)\cdot (q \left(2 q^3-4 q^2+1\right))
$$ 
This polynomial
is non-increasing on $[0,\frac{1}{2}]$ and remains positive 
on this interval. 
\end{proof}

\subsection{Profit lag.}
As before, we denote by $\Delta$ the difference of the average  revenue between selfish mining 
and mining honestly from the beginning.
\begin{proposition}
After $n$ attack cycles, we have 
$\Delta = \rho\cdot\left(1-\frac{1}{\delta}\right)n\, n_0 \tau_0= q^2\,n\, n_0 b$
and $\Delta$ stays constant during the BCH mining phase.
\end{proposition}
\begin{proof}
After a first phase of mining, we have $\Delta = 0$ since 
the revenue ratio on BCH and BTC are assumed to be equal.
At the end of the second phase, we have $\Delta = \left(1-\frac{1}{\delta}\right)
\rho\cdot n_0\tau_0$. Indeed, the second phase lasts $\frac{n_0\tau_0}{\delta}$
and the revenue ratio of the attacker during this phase is $\rho\delta$. 
Now, we use that $\rho = q\frac{b}{\tau_0}$ and
$\delta = \frac{1}{p}$ because only honest miners are mining during Phase $1$.
The strategy is then a repetion of alterning phases $1$ and $2$ and the result follows.
\end{proof}

We plot the graph of $\Delta$ in Figure 4.

\begin{figure}[!ht]
   \includegraphics[height=5cm, width=9cm]{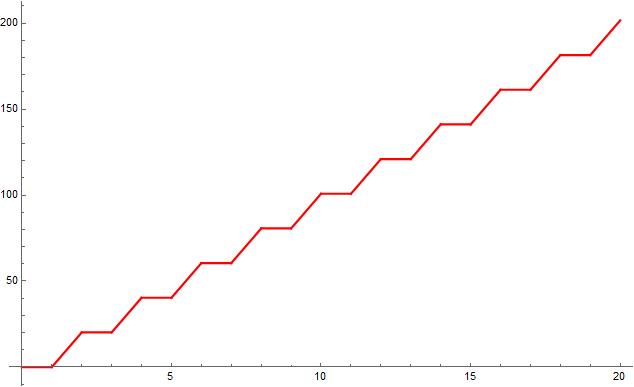}
   \caption{
    \footnotesize{Difference of average revenue between alternate and honest mining for $q=0.1$.
$X$-axis: progression of the official blockchain, in difficulty adjustment units. 
   $Y$-axis: revenue of the miner in coinbase units. 
      The strategy is definitely profitable after one difficulty adjustment.
   }}
\end{figure}

\section{Conclusion}
We revisit the selfish and intermitent mining strategies, and define alternate mining strategy 
for pairs of networks with the same PoW. 
We define the notion of ``profit lag''. We compute exact formulas for different 
strategies of the profit lag and the revenue ratio. 
We correct misunderstandings and unfounded claims in \cite{NRS20} by clarifying the 
profitability setup using the new notion of ``profit lag''. 
We show that, under natural hypothesis, 
the alternate mining strategy is the best one: its revenue ratio is the largest 
and the profit lag is the least. So, the existence of two networks sharing the same PoW algorithm can lead to instabilities of the difficulty and with blocks validated slower than normal.

\end{document}